\documentclass[11pt]{article}
\usepackage{amssymb,amsmath}
\usepackage{epsfig}
\usepackage{hyperref}
\usepackage{fullpage}
\usepackage{mathpazo}

\newtheorem{theorem}{Theorem}[section]

\newtheorem{claim}[theorem]{Claim}

\newcommand{\qedsymb}{\hfill{\rule{2mm}{2mm}}}
\newenvironment{proof}[1][]{\begin{trivlist}
\item[\hspace{\labelsep}{\bf\noindent Proof#1:\/}] }{\qedsymb\end{trivlist}}

\def\R{\mathbb{R}}

\newcommand{\inpc}[2]{\langle{#1},{#2}\rangle} % inproduct, < , >
\newcommand{\Tr}{\mbox{\rm Tr}}
\def\01{\{0,1\}}
\newcommand{\topic}[1]{\paragraph{{#1}.}}

\newcommand\ket[1]{{ |{#1} \rangle }}
\newcommand\bra[1]{{ \langle {#1} | }}

\newcommand{\eps}{\varepsilon}
\renewcommand{\epsilon}{\varepsilon}

\begin{document}

\title{\bf Bell Violations through Independent Bases Games}

\author{
 Oded Regev \footnote{Blavatnik School of Computer Science, Tel Aviv University, and
CNRS, D{\'e}partement d'Informatique, {\'E}cole normale sup{\'e}rieure, Paris.
Supported by the Israel Science Foundation, by the Wolfson Family
  Charitable Trust, and by a European Research Council (ERC) Starting Grant.
   }}

%\date{}

\maketitle

\begin{abstract}
In a recent paper, Junge and Palazuelos presented two two-player games exhibiting
interesting properties. In their first game, entangled players can
perform notably better than classical players. The quantitative gap between the two cases
is remarkably large, especially as a function of the number of inputs to the players.
In their second game, entangled players can perform notably better than
players that are restricted to using a maximally entangled state (of arbitrary dimension).
This was the first game exhibiting such a behavior.
The analysis of both games is heavily based on non-trivial results from Banach space
theory and operator space theory. Here we present two games exhibiting a similar
behavior, but with proofs that are arguably
simpler, using elementary probabilistic techniques and standard quantum information
arguments. Our games also give better quantitative bounds.
\end{abstract}

\section{Introduction}

\topic{Bell violations}
One of the most remarkable predictions of quantum mechanics, dating back to the mid-20th century~\cite{epr,bell:epr},
is the fact that remote parties that share entanglement can exhibit behaviors that cannot be explained classically.
Experimental verifications of these predictions are the strongest proof
we have that nature does not behave according to classical physics.

A recent line of work has focused on a quantitative analysis of this phenomenon
through the framework of two-player games. The setting is as follows.
Two non-communicating parties, called Alice and Bob, receive inputs $x$ and $y$ according
to some fixed and known probability distribution, and are required to produce outputs $a$ and $b$, respectively.
For each pair of inputs $(x,y)$ and pair of outputs $(a,b)$ there is a known probability
of winning the game. The goal of the players is to maximize their bias,
defined as the absolute value of the difference between the winning probability and
$1/2$.\footnote{One can also consider the goal of simply maximizing
the winning probability, a setting corresponding to what one might call \emph{positive Bell functionals}.
Unfortunately, we do not know if our games can be modified to apply
to this stronger setting. See~\cite{junge&palazuelos:largeviolation} for some results
on positive Bell functionals.}

The players can be classical, in which case their strategy is described by a function
mapping each input $x$ of Alice to some output $a$, and similarly a function for Bob.
They can also be entangled, in which case a strategy is described by an arbitrary
entangled state shared between the parties, together with a measurement for each input $x$ of Alice,
and a measurement for each input $y$ of Bob.
On inputs $x,y$, the players apply the measurements corresponding to $x$ and $y$ to the entangled state,
producing classical outputs $a$ and $b$, respectively.

The remarkable phenomenon mentioned above corresponds to the fact that there are
games where entangled players can achieve a strictly larger bias
than what classical players can achieve.
The CHSH game is one particularly famous example~\cite{chsh}.
Here the inputs $x\in\01$ and $y\in\01$ are uniformly
distributed, and Alice and Bob win the game if their respective outputs $a\in\01$ and $b\in\01$
satisfy $a\oplus b=x\wedge y$; in other words, $a$ should equal $b$ unless $x=y=1$.
It is easy to see that classical players can achieve a bias of at most $1/4$.
Entangled players, on the other hand, can achieve a bias of $\cos(\pi/8)^2 - 1/2 \approx 0.35$.
This is achieved already with $2$-dimensional entanglement (namely, one EPR-pair).

The ratio between the best bias achievable by entangled players and that achievable
by classical players is known as the \emph{Bell violation exhibited by the game}.
(See \cite{BuhrmanRSW11} for the origin of this term, as a well as a further discussion of the model.)
This ratio has been the subject of intensive study in recent years. Let us
mention a few results in this line of work. It follows from Grothendieck's
inequality that XOR games (i.e., games where the outputs are taken from the
set $\{0,1\}$ and the winning probability depends solely on the inputs and the XOR of the two
outputs, as is the case in the CHSH game) cannot exhibit a violation greater than some absolute constant~\cite{chtw:nonlocal}.
More generally, Junge and Palazuelos~\cite[Theorem~14]{junge&palazuelos:largeviolation} proved that if
Alice and Bob have at most $k$ possible outputs each, then the maximum violation obtainable
is at most $O(k)$.\footnote{This improved an earlier $O(k^2)$ upper bound due to Degorre et al.~\cite{dklr:nonsignal},
and was also obtained for a special case by Dukaric~\cite[Theorem~4]{dukaric:norm}.}
They also show that if the number of inputs to each player is at most $n$,
the maximum violation is at most $O(n)$. Finally, Junge et al.~\cite{PerezGarcia09arxiv}
showed that the maximum violation obtainable with local entanglement dimension $d$ is at most $O(d)$.

It is interesting to look for games that exhibit a violation
that is as close as possible to these upper bounds.
One of the strongest results in this direction is by
Junge and Palazuelos~\cite{junge&palazuelos:largeviolation} who (improving on the earlier~\cite{PerezGarcia09arxiv})
showed that for all $n$ there exists a game, which we call the JP game, exhibiting a violation of order $\sqrt{n}/\log n$
with $n$ inputs to each player, $n$ possible outputs to each player, and using entanglement
of local dimension $n$. This violation is only quadratically away from the known upper bounds in
all three parameters of the problem, and as such is quite strong.

In another recent result, Buhrman et al.~\cite{BuhrmanRSW11} presented a game,
called the Khot-Vishnoi (KV) game, which exhibits a violation of order $n/(\log n)^2$, with $2^n/n$ inputs,
$n$ outputs, and using entanglement of local dimension $n$.\footnote{They also
present the so-called Hidden Matching game, which we ignore here since it is weaker in
essentially all parameters.}
This violation is nearly optimal in terms of the number of outputs and in terms of the
local dimension, but not in terms of the number of inputs.
One nice feature of this game is that it is fully explicit;
the JP game, in contrast, is randomly constructed and only guaranteed
to exhibit the gap with high probability. Another advantage
of the KV game is that its analysis is relatively simple.
(Yet another advantage is that the violation can be obtained in terms of the
maximum winning probability, as opposed to maximum biases.)

\topic{First result: a direct and slightly improved analysis of the JP game}
The analysis of the JP game in~\cite{junge&palazuelos:largeviolation} (as well as the earlier
analysis in~\cite{PerezGarcia09arxiv}) are based on non-trivial tools from Banach space theory and
operator space theory, and as such are quite difficult to follow for readers
without the appropriate background. Our first contribution, presented in Section~\ref{sec:jungepala}, is an
alternative direct analysis of their game, using nothing more than basic
probabilistic arguments. Due to the more direct
analysis, we also obtain slightly better quantitative parameters, namely,
we reduce the number of inputs from $n$ to $n/\log n$.

\topic{Second result: limitation of maximally entangled states}
It is quite natural to expect that a maximally entangled state is always at least
as good a starting point for an entangled strategy as a ``less entangled'' state in the same dimension.
This is especially natural considering that almost all the entangled strategies one encounters in the literature
use a maximally entangled state. Somewhat surprisingly, this intuition is completely false:
Junge and Palazuelos~\cite{junge&palazuelos:largeviolation} presented a game
in which players using a certain non-maximally entangled state can do strictly better
than players using a maximally entangled state of arbitrarily high dimension.
More precisely, they considered a modification of the JP game in which the number of inputs
is increased to $2^{O(n^2)}$ and the number of outputs remains $n$.
They showed that using a certain non-maximally entangled state of dimension $n$,
entangled players can obtain a bias that is greater by a factor of order $\sqrt{n}/\log n$
than what can be achieved using a maximally entangled state of arbitrarily high dimension.
The analysis is heavily based on results from operator space theory.

In Section~\ref{sec:maxent} we present a game with a similar behavior.
Our analysis is based on a basic quantum information technique, namely,
the quantum random access code argument, and should hopefully be more familiar to
most readers.
Our game also happens to provide better quantitative performance. Namely,
we obtain the same gap of $\sqrt{n}/\log n$ using only $2^n/n$ inputs
(and still using $n$ outputs and $n$-dimensional entanglement).
Finally, unlike the JP game, our game is explicit.

We note that recently two other explicit games were presented
where the maximally entangled state is not optimal, one by
Liang, Vertesi, and Brunner~\cite{LiangVB10}, and one by
Vidick and Wehner~\cite{VidickW10}.
These games have the advantage of being very small (namely,
just two inputs and two outputs in the former game and three inputs
and two outputs in the latter game) but as a result the gap they
obtain is also quite small (and fixed, unlike our gap which grows to infinity with
the size of the game).

\topic{Discussion}
Although our second game differs from the JP game in its technical details,
the idea underlying both games is the same. Namely, in both games,
the inputs to Alice and Bob are (essentially) orthonormal bases of $\R^n$,
chosen independently of each other; each player outputs one vector from the input basis
with the goal being to pick vectors that are as close as possible to each other.
Given this strong similarity, it is natural to ask whether the JP game,
with its significantly smaller number of inputs, is already enough
to demonstrate the limitation of maximally entangled states.
Another interesting question is how good maximally entangled states are when one is interested
in maximizing the winning probability, as opposed to the bias. One result in this direction~\cite[Theorem 10]{junge&palazuelos:largeviolation}
shows that one can replace an arbitrary entangled state of dimension $n$ by a maximally entangled state,
reducing the winning probability by a factor of at most $O(\log n)$; it would be nice to obtain a similar statement
with a loss that is independent of the dimension.

\section{A Bell Inequality Violation}\label{sec:jungepala}

In this section we give a self-contained description and analysis of the Junge-Palazuelos game~\cite{junge&palazuelos:largeviolation}.
The main differences compared to~\cite{junge&palazuelos:largeviolation} are the following.
First and most importantly, instead of using Banach space theory to bound the classical winning probability of the game
we use elementary probabilistic arguments based on a union bound and a Chernoff-like tail bound.
(We do not know if the proof in~\cite{junge&palazuelos:largeviolation},
when unwrapped from its Banach space language, boils down to a similar probabilistic proof.) Our analysis
of the entangled strategy is similar to the original analysis in~\cite{junge&palazuelos:largeviolation}.
Second, to simplify the discussion a bit, we choose
to work with Gaussian variables instead of Bernoulli variables, a possibility already
mentioned in~\cite{junge&palazuelos:largeviolation}. Finally, we also choose to describe things in terms
of games, as opposed to Bell functionals (see~\cite{BuhrmanRSW11}), in order to work in a slightly more familiar setting.

We start with an informal description of the game. The game is parameterized by two integers $n$ and $k$.
The number of inputs to each player is $n$ and the number of outputs of each player is $k$.
For each input $x \in [n]$ to Alice we associate a uniformly chosen orthonormal basis of
$\R^k$, $u^{x}_1,\ldots,u^{x}_k$.
(For technical reasons, in the actual game these vectors won't
be quite orthogonal, but only mostly so.) Similarly, for each input $y \in [n]$ to Bob we associate
a uniformly chosen orthonormal basis of $\R^k$, $v^{y}_1,\ldots,v^{y}_k$. Recall that
the inner product between two uniformly distributed $k$-dimensional unit vectors behaves roughly like a normal distribution
with standard deviation $1/\sqrt{k}$. The game is the following. The players receive inputs $x,y$
chosen uniformly, and are supposed to return outputs $a,b$ for which $\inpc{u^{x}_a}{v^{y}_b}$
is positive. More precisely, we will set the winning probability of outputs $a,b$ to be
$\frac{1}{2}+\delta \inpc{u^{x}_a}{v^{y}_b}$, where the normalizing factor
$\delta \approx \sqrt{k}$ is chosen so that this value is between $0$ and~$1$.
(We note that for convenience, in the actual game we use a different normalization.)

An apparently good entangled strategy for the game follows almost immediately from the definition.
Using the maximally entangled state, Alice performs a measurement in the basis corresponding to her input $x$, and similarly
Bob performs a measurement in the basis corresponding to his input $y$.
It is not difficult to see that this makes the probability of
obtaining outputs $a,b$ proportional to the square of the inner product between them.
But this is not helpful: recall that our goal is to output pairs where the inner product itself is large,
and not its absolute value.
This can be solved by slightly modifying the quantum measurements and, somewhat
surprisingly, seems to require replacing the maximally entangled state (as used
in almost all other known entangled strategies) with a certain non-maximally entangled state.

It is interesting to note that there are strong similarities between the JP game and the KV game of~\cite{BuhrmanRSW11}.
For instance, both games are designed so that each input is naturally associated with
an orthonormal basis, which in turn leads to an entangled strategy. Roughly speaking, the main difference
between the games is that in the KV game, input pairs correspond to bases that are somewhat aligned and the outputs
are supposed to agree with the matching between the bases (so there are $k$ valid output pairs), whereas here input
pairs correspond to two independent random bases and the outputs are supposed to be on somewhat close
vectors (so roughly speaking, there are $k^2/2$ valid output pairs).

The main part of the proof consists of upper bounding the bias classical players can obtain. Here~\cite{junge&palazuelos:largeviolation}
use results from Banach space theory. We instead perform a relatively simple union bound over all
$k^{2n}$ classical strategies.

We now proceed with the formal description. We will need the following simple claims
on the Gaussian distribution.

\begin{claim}\label{clm:largedevgaussian}
For all $k \ge 1$ and $R \ge 2$, the probability that a standard $k$-dimensional Gaussian
variable has norm greater than $\sqrt{8k \ln R}$ is at most~$R^{-k}$.
\end{claim}
\begin{proof}
Consider the density function of the standard $k$-dimensional Gaussian, $(2
\pi)^{-k/2} e^{- \|x\|^2 / 2}$ as well as that of a Gaussian with standard deviation $2$
in each coordinate, $2^{-k} (2 \pi)^{-k/2} e^{- \|x\|^2 / 8}$. The ratio between these two densities is
$2^{-k} e^{3\|x\|^2/8}$. In particular, for any point of norm greater than $\sqrt{8 k
\ln R}$, the latter density is greater than the former by a factor of at least $R^k$.
The claim follows from the fact that the probability of obtaining a point of norm greater
than $R$ in the latter distribution is clearly at most 1.
\end{proof}

\begin{claim}\label{clm:gaussinp}
For all $k \ge 1$ and $\eps>0$,  the probability that the inner product of two standard $k$-dimensional Gaussians
is bigger than $\sqrt{16 \ln(4/\eps) \max(k,\ln(2/\eps))}$ in absolute value is at most~$\eps$.
\end{claim}

\begin{proof}
By rotational symmetry of the Gaussian distribution, we may assume one of the two vectors (call it $Y$)
is a scalar multiple of $e_1$, where that scalar is the vector's norm $\|Y\|$.
Hence the distribution we are considering is equal to that of $X \|Y\|$
where $X$ is a standard (one-dimensional) normal variable and $Y$ is an independent standard
$k$-dimensional Gaussian vector. By standard estimates, the probability that $X$ is greater
than $\sqrt{2\ln(4/\eps)}$ in absolute value is at most $\eps/2$.
By Claim~\ref{clm:largedevgaussian}, the probability that $\|Y\|$ is greater
than $\sqrt{8k \max(1,\ln(2/\eps) / k) }$ is also at most $\eps/2$. The claim follows.
\end{proof}

\medskip

\topic{The game}
We continue with a formal description of the game.
The game is parameterized by two integers $k$ and $n$, and we assume
for simplicity that $\sqrt{k} \le n \le 2^k$ (this can easily be extended to other settings).
As mentioned above, the number of inputs to each player is $n$;
for technical reasons, however, we add a $k+1$st output, which we think of as a ``pass'' output,
so the number of outputs of each player is $k+1$.
For each input $x \in [n]$ to Alice we independently choose $k$ standard Gaussian vectors
in $\R^k$, $u^{x}_1,\ldots,u^{x}_k$. We similarly define vectors
$v^{y}_1,\ldots,v^{y}_k$ for each of Bob's inputs.
The players receive inputs $x,y$ chosen uniformly. If any of the players outputs
$k+1$, they win with probability exactly $1/2$.
Otherwise, the probability that outputs $a,b\in[k]$ win (on input pair $x,y$)
is set to $\frac{1}{2}+\delta \inpc{u^{x}_a}{v^{y}_b}$, where $\delta := c/\sqrt{k \log n}$
for some constant $c>0$.
It follows from Claim~\ref{clm:gaussinp} and a union bound that if we choose $c>0$ small enough,
then with probability very close to $1$ (say, $1-n^{-10}$),
all these $n^2 k^2$ numbers are between $0$ and $1$, and
hence the game is well defined.\footnote{If one is content with a
Bell functional then this step can be avoided, and there is no need for $\delta$:
just define the functional as $M^{ab}_{xy} = \inpc{u^{x}_a}{v^{y}_b}$ or $0$ if
either $a$ or $b$ are $k+1$. The rest of the analysis is essentially identical.}
In the following we will show that the
classical winning probability of the game can differ from $1/2$ by at most
$O(\delta \max(\sqrt{(k\log k)/{n}}, \log k))$,\footnote{We note that
this bound is essentially tight in the interesting regime of $n \ge k/\log k$.
Just have Alice and Bob choose the vector whose first coordinate is largest.
Since among $k$ independent standard normal variables there is with high probability
at least one that is $\Omega(\sqrt{\log k})$, the inner product between Alice's
and Bob's vector has expectation of order $\log k$, leading to a winning probability
of $\frac{1}{2}+\Omega(\delta \log k)$.}
and that a strategy using entanglement of local dimension $k+1$ can win the game with probability at least
$\frac{1}{2}+ \Omega(\delta \sqrt{k})$. As a result we obtain a Bell violation
of $\Omega(\sqrt{k}/\log k)$ with $k+1$ possible outputs, local dimension $k+1$,
and a number of inputs $n$ that can be taken to be as small as $k/\log k$.
We note that this slightly improves on the analysis in~\cite{junge&palazuelos:largeviolation}
where the number of inputs was taken to be $n=k$,
leading to a violation of $\Omega(\sqrt{n}/\log n)$ compared with our
$\Omega(\sqrt{n/\log n})$.

\medskip

\topic{Classical bias}
As mentioned above, we bound the classical bias by performing a union
bound over all $(k+1)^{2n}$ classical strategies. Fix some arbitrary strategy given by functions $a,b:[n]\rightarrow[k+1]$.
The winning probability of this strategy is
$$  \frac{1}{2} + \frac{\delta}{n^2} \sum_{x,y} \inpc{u^{x}_{a(x)}}{v^{y}_{b(y)}} =
     \frac{1}{2} + \frac{\delta}{n^2}  \left\langle \sum_x u^{x}_{a(x)},\sum_y v^{y}_{b(y)} \right\rangle,$$
where the sums run over $x$ such that $a(x) \neq k+1$ and $y$ such that $b(y) \neq k+1$.
Notice now that $\sum_x u^{x}_{a(x)}$ and $\sum_y v^{y}_{b(y)}$ are two independent $k$-dimensional Gaussians
with standard deviation at most $\sqrt{n}$ in each coordinate (since they are sums of at most $n$ independent standard Gaussians).
Hence, by Claim~\ref{clm:gaussinp}, we know that the probability that the above inner product is greater than
$$ (\sqrt{n})^2 \cdot C \sqrt{ n \log k \max(k,n \log k)} = C n^2 \max\left(\sqrt{\frac{k\log k}{n}}, \log k\right)$$
in absolute value is much less than $(k+1)^{-2n}$,
where $C>0$ is a large enough constant.
By the union bound we conclude that the above-mentioned bound $1/2 \pm O(\delta \max(\sqrt{(k\log k)/{n}}, \log k))$
on the classical winning probability indeed holds with high probability
(probability taken over the choice of the vectors $u^{x}_1,\ldots,u^{x}_k$ and $v^{y}_1,\ldots,v^{y}_k$).

\medskip

\topic{Entangled strategy}
We use the following (non-maximally) entangled state with local dimension $k+1$,
\begin{equation}\label{eq:nonmaxentjp}
\ket{\psi} := \frac{1}{\sqrt{2k}} \sum_{i=1}^{k} \ket{i,i} + \frac{1}{\sqrt{2}} \ket{k+1,k+1} .
\end{equation}
Define vectors $\tilde{u}^{x}_a \in \R^{k+1}$ by taking ${u}^{x}_a$ and adding as the $k+1$st coordinate the value
$1$. For each input $x$ to Alice we create the POVM $\{A^{x}_a\}_a$ given by
$$ A^{x}_a := \frac{1}{10k} \ket{ \tilde{u}^{x}_a } \bra{ \tilde{u}^{x}_a },$$
and $A^{x}_{k+1} = I - \sum_{a=1}^k A^{x}_a$. We similarly define POVM measurements for Bob,
$\{B^{y}_b\}_b$.
We claim that with high probability the largest eigenvalue of $\sum_{a=1}^k A^{x}_a$ is at most $1$,
and hence these are really well-defined POVMs. This holds,
roughly speaking, because the vectors $\tilde{u}^{x}_a$ are more or less orthogonal
and of norm $\approx \sqrt{k}$, and the normalization by $10k$ is more than sufficient.
More precisely, let $Z$ be a $(k+1) \times k$ matrix whose columns are
given by $\tilde{u}^{x}_1,\ldots,\tilde{u}^{x}_k$; then our goal is
to show that the largest eigenvalue of $Z Z^t$ is at most $10k$ with high
probability. For this, one needs to recall that the largest singular value of a
$k \times k$ matrix whose entries are i.i.d.\ standard Gaussian is tightly concentrated
around $2 \sqrt{k}$ (see, e.g.,~\cite[Corollary 35]{VershyninNotes}); by (say) the triangle inequality for the operator norm, the addition of an extra row of ones
in $Z$ can increase this by at most $\sqrt{k}$, so in particular the probability
that the largest singular value of $Z$ is greater than $\sqrt{10 k } > 3 \sqrt{k}$
is very small. This implies that with high probability the largest eigenvalue
of $Z Z^t$ is at most $10k$, as required.

It remains to calculate the winning probability of this strategy. Given inputs $x,y$, the probability
of the players producing outputs $a,b\in[k]$ is
\begin{align*}
\bra{\psi} A^{x}_a \otimes B^{y}_b \ket{\psi} &=
  \frac{1}{100k^2} \left( \frac{1}{\sqrt{2k}} \sum_{i=1}^k (\tilde{u}^{x}_a)_i (\tilde{v}^{y}_b)_i + \frac{1}{\sqrt{2}}
  (\tilde{u}^{x}_a)_{k+1} (\tilde{v}^{y}_b)_{k+1} \right)^2 \\
  &= \frac{1}{100k^2} \left( \frac{1}{\sqrt{2k}} \inpc{u^{x}_a}{v^{y}_b} + \frac{1}{\sqrt{2}} \right)^2.
\end{align*}
The winning probability is therefore
\begin{align*}
 \frac{1}{2}+ \frac{\delta}{100k^2 n^2} \sum_{x,y=1}^n \sum_{a,b=1}^k
    \inpc{u^{x}_a}{v^{y}_b} \left( \frac{1}{\sqrt{2k}} \inpc{u^{x}_a}{v^{y}_b} + \frac{1}{\sqrt{2}} \right)^2.
\end{align*}
Since $\inpc{u^{x}_a}{v^{y}_b}$ is distributed roughly like a normal variable with standard deviation
$\sqrt{k}$, the expectation of the expression inside the sum should behave like $\Theta(\sqrt{k})$.
Indeed, by expanding the inner products one easily sees that the expectation is exactly $\sqrt{k}$.
Hence the expectation of the winning probability is $\frac{1}{2}+\Theta(\delta \sqrt{k})$.
Furthermore, we claim that the winning probability is close to this value with high probability,
and not just in expectation.
One straightforward (yet tedious) way to see this is to compute the standard deviation of the above expression,
which turns out to be $O(\delta (nk)^{-1/2})$,
and then use Chebyshev's inequality.\footnote{In fact,
since the winning probability is a low-degree polynomial in normal variables,
concentration results for Gaussian polynomials show that
this probability is exponentially to $1$ (see, e.g.,~\cite{Latala06} and references therein).}

\smallskip
Finally, by a union bound, we obtain that there is a nonzero (and in fact high) probability
that a game chosen as above satisfies all three requirements simultaneously (namely,
it is well-defined, the bias achievable by classical players is low, and the bias achievable
by entangled players is high), and we are done.

\section{Maximally Entangled States are not Always the Best}\label{sec:maxent}

\topic{The game}
The game we consider can be seen as a natural hybrid between the Junge-Palazuelos game
and the KV game of~\cite{BuhrmanRSW11}. Namely, the inputs and predicates are essentially those in
the Junge-Palazuelos game, except we use the explicit hypercube structure as in the
KV game. Details follow. The game is parameterized by $n$, which is a power of two.
We consider the group $\{0,1\}^n$ of all $n$-bit strings with bitwise addition modulo 2, and let
$H$ be the subgroup containing the $n$ Hadamard codewords. This subgroup partitions
$\{0,1\}^n$ into $2^n/n$ cosets of $n$ elements each. The players receive two independent
and uniform cosets of $H$, call them $H_A$ and $H_B$. Alice outputs an element
$a \in H_A$ and similarly Bob outputs an element $b \in H_B$. Given outputs $a,b$, we
define their probability of winning the game as
$$1 - \frac{1}{n} d(a,b),$$
where $d(a,b)$ denotes the Hamming distance between $a$ and $b$.
Note that the number of possible inputs to each player is $2^n/n$
and that the number of outputs is $n$.

The similarity with the JP game can be seen by thinking of elements of the hypercube
as unit vectors in $\{-1/\sqrt{n},1/\sqrt{n}\}^n$ instead of in $\{0,1\}^n$. Then
each coset of $H$ defines an orthonormal basis, and the winning probability
above is (up to normalization) the same as in the JP game.

\medskip

\topic{Entangled strategy}
We first claim that entangled players can obtain a winning probability of $1/2+\Omega(1/\sqrt{n})$
using an entangled state of local dimension $n+1$. We omit the proof since the strategy that obtains this is
essentially identical to the one used in the JP game (see Section~\ref{sec:jungepala}). The only
thing to notice is that a ``pass" output is not needed; instead, a player can output a uniform
output, thereby guaranteeing a winning probability that is within $\pm 1/n$ of $1/2$ no matter what the other
player outputs, which is essentially as good as a pass output.

\medskip

\topic{Upper bound on the bias achievable with maximally entangled states}
Recall that the entangled strategy above crucially relies on a non-maximally entangled state (Eq.~\eqref{eq:nonmaxentjp}).
In the following we will show that the winning probability of any strategy using a maximally entangled state
of arbitrarily high dimension cannot differ from $1/2$ by more than $O(\log(n)/n)$
(this of course also applies to classical strategies).
Hence we obtain a gap of $\sqrt{n}/\log n$ between entangled strategies and
entangled strategies restricted to using a maximally entangled state.
As mentioned above, the same gap was shown in~\cite{junge&palazuelos:largeviolation},
but using a non-explicit construction and with a much larger number of inputs, namely,
$2^{n^2}$. Moreover, their proof is heavily based on operator space theory whereas ours
only uses basic notions in quantum information theory. It is interesting to point out that already classical strategies can obtain
a winning probability of $1/2 + \Omega(\log(n)/n)$ hence
entangled strategies using the maximally entangled state do not offer any significant advantage
over classical strategies. (One classical strategy that achieves this is the one where Alice and Bob output
the element of their coset of largest Hamming weight.)

Consider an arbitrary entangled strategy using a maximally
entangled state. First, we assume that for any input $H_A$ to Alice, her output
distribution is uniformly distributed over $H_A$, and similarly for Bob. This assumption
is without loss of generality since given an arbitrary strategy, we can have Alice and Bob
use shared randomness (which can be obtained from a shared maximally entangled state)
to pick a uniform element $r \in H$ and then add it to whatever they were about to output.
This clearly does not affect their winning probability and provides the uniform marginal
property.

Assume the strategy uses a maximally entangled state of an arbitrarily large local dimension $D$,
$\ket{\psi} = D^{-1/2} \sum_{i=1}^D \ket{ii}$.
Recall that such a strategy
is described by a POVM $\{E_a\}_{a \in H_A}$ for each input $H_A$ to Alice,
and similarly $\{F_b\}_{b \in H_B}$ for Bob. Here, $E_a$ are positive semidefinite
$D \times D$ matrices satisfying $\sum_{a \in H_A} E_a = I$ for all $H_A$, and similarly
for $F_b$. The probability of obtaining outputs $a,b$ on inputs $H_A,H_B$ is
$$ \bra{\psi} E_a \otimes F_b \ket{\psi} = \frac{1}{D} \, \Tr(E_a F_b).$$
Hence the winning probability is
\begin{align*}
&\sum_{a,b \in \{0,1\}^n} \left(\frac{n}{2^n}\right)^2 \frac{1}{D} \, \Tr(E_a F_b) \cdot (1-d(a,b)/n) = \\
&\qquad \frac{1}{2} + \frac{1}{2n} \sum_{a,b \in \{0,1\}^n} \left(\frac{n}{2^n}\right)^2 \frac{1}{D} \, \Tr(E_a F_b) \cdot (n-2d(a,b)).
\end{align*}
Notice also that by assumption on the uniform marginals, $\Tr(E_a) = \Tr(F_b) = D/n$ for all $a,b$.

For $1 \le i \le n$ define
$$ \hat{E}_i := \frac{n}{2^n} \left( \sum_{a|a_i=0} E_a - \sum_{a|a_i=1} E_a \right),$$
and define $\hat{F}_i$ similarly. Then it is not difficult to see that the winning probability
can be written as
$$ \frac{1}{2} + \frac{1}{2n} \sum_{i=1}^n \frac{1}{D}\, \Tr(\hat{E}_i \hat{F}_i).$$
(We note that $\hat{E}_i$ is up to normalization the Fourier coefficient at location $\{i\}$ of the matrix-valued
function $a \mapsto E_a$ and that this expression for the winning probability arises
naturally when one expresses the winning probability in terms of the Fourier transform.)
Now, by Cauchy-Schwarz, we see that the distance from $1/2$ of the winning probability is at most
$$ \frac{1}{2n} \left( \sum_{i=1}^n \frac{1}{D} \, \Tr(\hat{E}_i^2)\right)^{1/2}
    \left( \sum_{i=1}^n \frac{1}{D} \, \Tr(\hat{F}_i^2)\right)^{1/2}.$$
In the following we will show that for any $E_a$ as above,
$$ \sum_{i=1}^n \frac{1}{D} \, \Tr(\hat{E}_i^2) \le O(\log n).$$
Since the same bound also holds for $\hat{F}_i$, we will be done.

This bound can be interpreted as a bound on the weight in the first Fourier level of
the matrix-valued function $a \mapsto E_a$ and as such one can try to apply
the hypercontractive inequality or the
matrix-valued version thereof, as in~\cite{brw:hypercontractive}.
While this provides the desired bound for classical strategies, it seems
to be not strong enough for entangled strategies. Instead,
we use an argument based on the von Neumann entropy, similar to the one
used by Nayak~\cite{nayak:qfa} in the context of so-called quantum random access codes.

For $a \in \{0,1\}^n$ define $\sigma_a$ to be the quantum state on a $D$-dimensional
system whose density matrix is $\frac{n}{D} E_a$. Notice that the latter is positive
semidefinite and of trace $1$, as required. Notice also that its highest eigenvalue
is at most $n/D$ and that the uniform mixture of all $\sigma_a$ is the completely
mixed state, i.e., $2^{-n} \sum_{a \in \{0,1\}^n} \sigma_a = I/D$.
We also define
$$ \hat{\sigma}_i = \frac{1}{2^n} \left( \sum_{a|a_i=0} \sigma_a - \sum_{a|a_i=1} \sigma_a \right),$$
and notice that its eigenvalues are between $-1/D$ and $1/D$. Our goal can be equivalently
written as
$$ \sum_{i=1}^n D \, \Tr(\hat{\sigma}_i^2) \le O(\log n).$$

Let $A$ and $M$ be two quantum systems of dimensions $2^n$ and $D$, respectively, in the
joint state
$$\frac{1}{2^n} \sum_{a \in \01^n} \ket{a}\bra{a} \otimes \sigma_a.$$
Note that the reduced state of $M$ is completely mixed, and hence its von Neumann
entropy is $S(M)= \log D$. The conditional von Neumann entropy $S(M|A)$, which
is simply the average entropy of $\sigma_a$, is at least $\log(D/n)$ by the bound
on the highest eigenvalue of $\sigma_a$. Hence, the mutual information of $A$ and $M$
satisfies
$$ S(A:M) = S(M) - S(M|A) \le \log n.$$
On the other hand,
$$ S(A:M) = \sum_{i=1}^n S(A_i : M | A_1\ldots A_{i-1}) \ge \sum_{i=1}^n S(A_i : M).$$
The equality is a chain rule and follows immediately from the definitions. The inequality
follows immediately from the strong subadditivity of the von Neumann entropy (which is a deep statement),
and the fact that $A_1,\ldots,A_n$ are independent. Now,
$$ S(A_i:M) = S(M) - S(M|A_i) = \log D - \frac{1}{2} (S(M|A_i=0) + S(M|A_i=1)).$$
Let $-\frac{1}{D} \le \lambda_1,\ldots,\lambda_D \le \frac{1}{D}$ be the eigenvalues of $\hat{\sigma}_i$.
Then since the density matrix of $M|A_i=0$ is $I/D + \hat{\sigma}_i$
and that of $M|A_i=1$ is $I/D - \hat{\sigma}_i$, the above is
\begin{align*}
&\log D - \frac{1}{2} \left(  - \sum_{j=1}^D \left(\frac{1}{D}+\lambda_j \right) \log \left(\frac{1}{D}+\lambda_j \right)
  - \sum_{j=1}^D \left(\frac{1}{D}-\lambda_j \right) \log \left(\frac{1}{D}-\lambda_j \right)  \right) =\\
&\quad \frac{1}{D} \sum_{j=1}^D \left( \log D + \left(\frac{1}{2}+\frac{D}{2}\lambda_j \right) \log \left(\frac{1}{D}+\lambda_j \right)
  + \left(\frac{1}{2}-\frac{D}{2} \lambda_j \right) \log \left(\frac{1}{D}-\lambda_j \right) \right)  =\\
& \quad \frac{1}{D} \sum_{j=1}^D \left( 1-H \left( \frac{1}{2} + \frac{D}{2} \lambda_j \right)\right),
\end{align*}
where $H$ is the binary entropy function, $H(p)=-p \log p- (1-p) \log (1-p)$.
Using the inequality $1-H(p) \ge \frac{2}{\ln 2} (p-\frac{1}{2})^2$, we have
$$ S(A_i:M) \ge \frac{2}{\ln 2} \, \frac{1}{D} \, \sum_{j=1}^D \left( \frac{D}{2} \lambda_j \right)^2 =
  \frac{1}{2 \ln 2} \, D \, \Tr(\hat{\sigma}_i^2).
$$
So by summing over all $i$ we obtain
$$ \sum_{i=1}^n D \, \Tr(\hat{\sigma}_i^2) \le 2 \ln 2 \cdot \log n,$$
as required.

\topic{Acknowledgments}
I thank Jop Bri\"et, Yeong-Cherng Liang, Carlos Palazuelos, and Thomas Vidick for their useful comments,
and Ronald de Wolf for extensive discussions throughout this project.

\newcommand{\etalchar}[1]{$^{#1}$}

\end{document}